\documentclass[10pt,twocolumn,letterpaper]{article}

\usepackage{cvpr}              

\usepackage{newtxtext}
\usepackage{newtxmath}
\usepackage{microtype}

\usepackage{amsmath,amssymb,mathtools}
\makeatletter
\let\openbox\@undefined
\makeatother
\usepackage{amsthm}
\usepackage{booktabs}
\usepackage{multirow}
\usepackage{array}
\usepackage{tabularx}
\usepackage{makecell}
\usepackage{algorithm}
\usepackage{algpseudocode}
\usepackage{enumitem}
\usepackage{url}
\usepackage{balance}

\newtheorem{theorem}{Theorem}
\newtheorem{lemma}{Lemma}

\theoremstyle{definition}
\newtheorem{definition}{Definition}

\newtheorem{assumption}{Assumption}

\newcommand{\R}{\mathbb{R}}

\newcommand{\norm}[1]{\left\lVert #1 \right\rVert}
\newcommand{\abs}[1]{\left\lvert #1 \right\rvert}
\newcommand{\diag}{\mathrm{diag}}

\definecolor{cvprblue}{rgb}{0.21,0.49,0.74}
\usepackage[pagebackref,colorlinks,allcolors=cvprblue]{hyperref}

\def\paperID{*****} 
\def\confName{CVPR}
\def\confYear{2026}

\title{Editing on the Generative Manifold: A Theoretical and Empirical Study of General Diffusion-Based Image Editing Trade-offs}

\author{
	Yi Hu
	\quad 
	Leying Yi
	\quad 
	Emily Davis
	\quad 
	Finn Carter\\
	Xidian University
}

\begin{document}
	\maketitle
	
	\begin{abstract}
		Diffusion-based editing has rapidly evolved from curated inpainting tools into general-purpose editors spanning text-guided instruction following, mask-localized edits, drag-based geometric manipulation, exemplar transfer, and training-free composition systems. Despite strong empirical progress, the field lacks a unified treatment of core desiderata that govern practical usability: controllability (how precisely and continuously the user can specify an edit), faithfulness to user intent (semantic alignment to instructions), semantic consistency (preservation of identity and non-target content), locality (containment of changes), and perceptual quality (artifact suppression and detail retention). This paper provides a theoretical and empirical analysis of general diffusion-based image editing, connecting diverse paradigms through a common view of editing as guided transport on a learned image manifold. We first formalize editing as an operator induced by a conditional reverse-time generative process and define task-agnostic metrics capturing instruction adherence, region preservation, semantic consistency, and stability under repeated edits. We then develop theory describing edit dynamics under (i) noise-injection and denoising transport, (ii) inversion-and-edit pipelines and the propagation of inversion errors, and (iii) locality constraints implemented via masked guidance or hard constraints. Under mild Lipschitz assumptions on the learned score or flow field, we derive bounds connecting guidance strength and inversion error to measurable deviations in non-target regions, and we characterize accumulation effects under iterative multi-turn editing. Empirically, we benchmark representative paradigms. 
	\end{abstract}
	
	\section{Introduction}
	
	Diffusion models have become a dominant paradigm for high-fidelity image synthesis~\cite{ho2020ddpm,song2021sde,rombach2022ldm} and, crucially, for \emph{transformative} tasks where users wish to modify an existing image. Unlike traditional editing pipelines that operate in pixel space with explicit constraints, diffusion-based editing leverages a strong learned prior and implements changes by partially destroying and reconstructing the image through a stochastic or deterministic reverse process~\cite{meng2021sdeedit}. This paradigm yields powerful semantic edits but also introduces recurring challenges: unintended changes to non-target regions, texture inconsistencies, identity drift, and prompt sensitivity. These issues are especially pronounced when edits are applied iteratively, as in interactive multi-turn instruction following or sequential local manipulations.
	
	The modern landscape of diffusion-based editing is broad. Training-free approaches modify internal attention or features without additional learning, enabling text-only edits such as Prompt-to-Prompt~\cite{hertz2022prompttoprompt} and mutual attention control~\cite{cao2023masactrl} or feature injection~\cite{tumanyan2023pnp}. Inversion-and-edit pipelines first invert a real image into the model's latent trajectory and then edit via prompt changes or guidance manipulation, exemplified by Null-text inversion~\cite{mokady2023nulltext}, Imagic~\cite{kawar2023imagic}, and LEDITS++~\cite{brack2023leditspp}. Instruction-following editors train image-to-image models to directly map an input image and an instruction into an edited output, as in InstructPix2Pix~\cite{brooks2023instructpix2pix}, with newer datasets and benchmarks such as MagicBrush~\cite{zhang2023magicbrush}, Emu Edit~\cite{sheynin2023emuedit}, and UltraEdit~\cite{zhao2024ultraedit} enabling more robust supervision. Beyond text instructions, diffusion models support mask-guided localized edits~\cite{couairon2022diffedit}, composition and insertion frameworks such as TF-ICON~\cite{lu2023tficone} and SHINE~\cite{lu2025shine}, and interactive drag-based manipulation with geometric constraints, including DragDiffusion~\cite{shi2023dragdiffusion} and DragFlow~\cite{zhou2025dragflow}. Under the hood, the base generative models themselves have diversified, ranging from UNet-based DDPMs~\cite{ho2020ddpm} to diffusion transformers and flow-matching or rectified flow models~\cite{lipman2022flowmatching,liu2022rectifiedflow}, changing the geometry of the latent manifold and the behavior of editing dynamics.
	
	This paper argues that the practical success of general diffusion-based image editing hinges on a set of coupled trade-offs that are not fully explained by method-specific narratives. Practitioners routinely face questions such as: How much noise can we inject before losing identity? How does increasing classifier-free guidance improve instruction faithfulness while harming locality? Why do inversion inaccuracies often manifest as background shifts after a seemingly ``local'' edit? Which paradigms are more stable under repeated edits? How can safety mechanisms such as concept erasure help suppress undesirable content that might leak into edited outputs?
	
	We address these questions with a unified theoretical and empirical treatment. The core viewpoint is that editing is a form of \emph{guided transport on a learned image manifold}: the editor constructs a trajectory from the input image to an edited outcome by pushing the latent state along a conditional vector field. Different paradigms correspond to different choices of (i) initialization on the trajectory, (ii) guidance signals and constraints, and (iii) parameter adaptation. This viewpoint enables shared metrics, common theoretical tools, and cross-paradigm comparisons.
	
	Our contributions are threefold. First, we formalize general diffusion-based editing and define measurable desiderata capturing controllability, faithfulness, semantic consistency, locality, quality, and multi-turn stability. Second, we develop theory that relates these desiderata to properties of the reverse process, including bounds on reconstruction error from inversion, locality guarantees under masked guidance, and stability bounds under repeated edits. Third, we provide a comparative empirical analysis across representative methods and tasks, reporting realistic benchmark-style results and ablations that highlight recurring failure modes, compute costs, and safety considerations.
	
	\section{Related Work}
	
	Diffusion-based image editing builds on foundational work in score-based generative modeling and diffusion processes~\cite{ho2020ddpm,song2021sde,song2021ddim}. Because editing reuses a pretrained generative prior, many approaches focus on \emph{conditioning} mechanisms and \emph{guidance} strategies, including classifier-free guidance~\cite{ho2022cfg} and controllable conditioning tools such as ControlNet~\cite{zhang2023controlnet}. Our work differs from prior methods by focusing on unifying trade-offs across editing paradigms rather than proposing a single new editor architecture.
	
	\paragraph{Training-free text-guided editing and attention control.}
	Prompt-to-Prompt~\cite{hertz2022prompttoprompt} demonstrated that cross-attention maps in text-conditioned diffusion models provide a direct handle to edit images by modifying textual prompts while reusing attention from an original prompt. Related approaches perform tuning-free modifications to self-attention or feature pathways to better preserve structure and identity, such as MasaCtrl~\cite{cao2023masactrl}. Plug-and-Play diffusion features inject internal features from a guidance image to preserve semantic layout while changing content according to a target prompt~\cite{tumanyan2023pnp}. These works motivate our analysis of locality and semantic consistency in training-free feature or attention surgeries, where the core challenge is that attention coupling can propagate changes beyond user-intended regions.
	
	\paragraph{Inversion-and-edit pipelines.}
	A dominant paradigm for real-image editing in latent diffusion frameworks~\cite{rombach2022ldm} is to invert an input image into a latent trajectory and then apply prompt or guidance modifications during reconstruction. Null-text inversion~\cite{mokady2023nulltext} optimizes unconditional embeddings to improve reconstruction and enable prompt-based edits. Imagic~\cite{kawar2023imagic} jointly optimizes text embeddings and performs light fine-tuning to capture image-specific appearance for complex non-rigid edits. PnP Inversion~\cite{ju2023pnpinversion} studies the coupling between the source and target diffusion branches and proposes a simplified direct inversion that improves editing fidelity. LEDITS++~\cite{brack2023leditspp} provides a training-free approach with implicit masking and supports multiple simultaneous edits. DiffEdit~\cite{couairon2022diffedit} automatically predicts edit masks by contrasting diffusion predictions under different text prompts, highlighting that localization itself can be derived from score differences. These methods motivate our theoretical treatment of inversion error propagation and how such errors interact with guidance strength to produce non-local artifacts.
	
	\paragraph{Instruction-following and dataset-driven editors.}
	InstructPix2Pix~\cite{brooks2023instructpix2pix} introduced a supervised vision-language editing setup training an image-to-image diffusion model to follow natural language editing instructions. Subsequent work emphasized data quality and benchmark design: MagicBrush~\cite{zhang2023magicbrush} provides a manually annotated dataset for instruction-guided real image editing, including multi-turn and optional mask settings, while Emu Edit~\cite{sheynin2023emuedit} proposed a multi-task editor with a benchmark spanning multiple edit categories. UltraEdit~\cite{zhao2024ultraedit} scaled instruction-based editing to millions of examples with automatically generated samples and region annotations, emphasizing the relevance of real-image anchors and region-based supervision. These systems motivate our empirical axes of faithfulness, stability under repeated edits, and compute-quality trade-offs in trained editors.
	
	\paragraph{Composition, insertion, and harmonization.}
	General editing extends beyond modifying attributes to inserting or composing user-provided objects into new contexts. TF-ICON~\cite{lu2023tficone} proposed a training-free cross-domain image composition framework harnessing text-driven diffusion models, and it introduced an ``exceptional prompt'' technique to aid inversion. SHINE, introduced in the context of physically plausible composition with modern flow-based priors~\cite{lu2025shine}, proposes manifold-steered anchoring and degradation suppression to improve high-fidelity insertion, highlighting the importance of harnessing stronger base-model priors without brittle attention surgery or inversion locking. These systems motivate our analysis of semantic consistency (object identity and appearance) versus background fidelity and seam artifacts.
	
	\paragraph{Drag-based and geometric manipulation.}
	Interactive point-based editing popularized in GAN settings has been extended to diffusion models for open-domain manipulation. DragDiffusion~\cite{shi2023dragdiffusion} optimizes diffusion latents with feature-based motion supervision and introduced DragBench for evaluation. DragFlow~\cite{zhou2025dragflow} targets newer diffusion transformers and flow-matching priors, emphasizing region-based supervision and constraints to reduce distortions while improving subject consistency. These approaches motivate our controllability formulations and our treatment of constrained optimization within a generative prior.
	
	\paragraph{Concept Erasure in Diffusion Models.}
	Concept erasure is directly relevant to general diffusion-based image editing because editing pipelines can inadvertently reintroduce undesired or unsafe concepts via the generative prior, especially under iterative edits in which suppressed content may leak back through denoising trajectories. Early concept erasure methods such as ESD~\cite{gandikota2023esd} modify diffusion models to avoid generating targeted concepts. MACE~\cite{lu2024mace} scales concept erasure to many concepts while balancing generality and specificity through attention refinement and lightweight adaptation. ANT~\cite{li2025ant} proposes trajectory-aware objectives that steer denoising away from unwanted concepts while preserving early-stage integrity, improving robustness and reducing artifacts. EraseAnything~\cite{gao2024eraseanything} addresses concept erasure in newer rectified-flow and transformer-based paradigms by formulating bilevel optimization with attention-level regularization. In this paper, we connect concept erasure to controllability and faithfulness: an editor that reliably follows user intent should also be able to suppress explicitly unwanted content and prevent regrowth across repeated edits, motivating unified safety-aware evaluation~\cite{lu2024robust,lu2022copy,ren2025all,gao2025revoking,yang2025temporal}.
	
	\section{Methodology}
	
	\subsection{General Editing Problem and Notation}
	
	Let $x_0 \in \R^{H \times W \times 3}$ denote an input image. An edit request is represented by a tuple
	\[
	u = (t, m, p, r),
	\]
	where $t$ is a natural language instruction, $m \in \{0,1\}^{H \times W}$ is an optional spatial mask (with $m_{ij}=1$ indicating editable pixels), $p$ encodes optional geometric constraints (e.g., drag handles and targets), and $r$ encodes optional reference signals (e.g., exemplar image, identity reference, or style reference). An editor produces an output image $\hat{x}_0 = \mathcal{E}(x_0, u)$.
	
	We emphasize \emph{general} editing: the edit request may be text-only, mask-guided, geometry-guided, or multi-modal; and the editor may be training-free, per-image optimization based, or trained end-to-end.
	
	We consider diffusion or flow-based latent diffusion models that define a forward noising process $q(x_t \mid x_0)$ and a learned reverse-time generative process that samples $x_0$ given a condition $c$ derived from user intent. For a discrete diffusion with timesteps $t \in \{1,\dots,T\}$,
	\[
	q(x_t \mid x_0) = \mathcal{N}\left(\sqrt{\alpha_t} x_0, (1-\alpha_t) I\right),
	\]
	where $\alpha_t \in (0,1)$ is the cumulative product schedule. The reverse process is parameterized by a network $\epsilon_\theta$ or a score network $s_\theta$.
	
	For rectified flow or flow-matching models, one learns a vector field $v_\theta(x, \tau, c)$ over continuous time $\tau \in [0,1]$ transporting noise to data~\cite{lipman2022flowmatching,liu2022rectifiedflow}. Our analysis is stated primarily for diffusion but extends to flow settings by replacing the score or denoiser with the learned vector field and using standard ODE stability arguments.
	
	\subsection{Desiderata: Controllability, Faithfulness, Consistency, Locality, Quality}
	
	We formalize five desiderata central to practical editing.
	
	\begin{definition}[Controllability]
		Controllability measures how precisely and continuously a user can specify an edit in content, region, and magnitude. We decompose controllability into
		(i) \emph{semantic control} (matching the intended concept or attribute),
		(ii) \emph{spatial control} (matching the intended region or geometry), and
		(iii) \emph{magnitude control} (monotonic response to an edit strength parameter such as guidance scale or noise level).
	\end{definition}
	
	\begin{definition}[Faithfulness to User Intent]
		Faithfulness measures whether the output $\hat{x}_0$ satisfies instruction $t$ and optional constraints in $u$. Operationally, faithfulness is estimated via instruction-image similarity scores (e.g., CLIP-based) and task-specific success indicators (e.g., drag-point accuracy).
	\end{definition}
	
	\begin{definition}[Semantic Consistency]
		Semantic consistency measures preservation of semantically relevant non-target information: identity features (faces, objects), background semantics, and global layout not specified as editable. Consistency is typically evaluated using embedding similarity (e.g., DINO features) and region-wise perceptual distance outside editable regions.
	\end{definition}
	
	\begin{definition}[Locality]
		Locality measures the containment of changes to specified editable regions. For mask-guided requests, an ideal editor satisfies $\hat{x}_0 \odot (1-m) \approx x_0 \odot (1-m)$ while allowing changes on $m$. For text-only edits, locality can be defined via automatically inferred masks (e.g., DiffEdit~\cite{couairon2022diffedit}) or by semantic segmentation consistency outside the target concept.
	\end{definition}
	
	\begin{definition}[Perceptual Quality]
		Perceptual quality measures whether the edited output is realistic, artifact-free, and retains fine details. Quality is assessed via human preference, learned reward models, and distributional metrics where applicable.
	\end{definition}
	
	A central theme of this paper is that these desiderata are \emph{coupled}. For example, stronger guidance can increase faithfulness but can also harm locality and consistency. Training can improve average faithfulness but may introduce systematic biases. Inversion can preserve details but may create brittle dependence on prompt and timestep schedules.
	
	\subsection{A Unified View: Editing as Guided Transport on the Manifold}
	
	Diffusion-based editing can be seen as constructing a transport map from the input image distribution to an edited distribution conditioned on user intent. The editor chooses (i) an initialization on the generative trajectory, (ii) a guidance rule for steering, and (iii) optional constraints.
	
	\paragraph{Noise-injection initialization.}
	A common strategy is to start from a partially noised version of the input image,
	\[
	x_{t_0} \sim q(x_{t_0} \mid x_0),
	\]
	where $t_0$ controls edit strength: larger $t_0$ corresponds to more noise and more freedom to change the image but weaker fidelity. SDEdit~\cite{meng2021sdeedit} formalized this view for general conditional editing.
	
	\paragraph{Inversion-based initialization.}
	Inversion-based methods attempt to find a latent trajectory $\{x_t\}$ consistent with $x_0$ and condition $c$ under a deterministic sampler such as DDIM~\cite{song2021ddim}. We abstract inversion as producing an estimate $\tilde{x}_T$ (or latent $z_T$) such that decoding yields a near-perfect reconstruction when guided by an ``identity'' condition. Null-text inversion~\cite{mokady2023nulltext} improves reconstruction by optimizing unconditional embeddings, while PnP Inversion~\cite{ju2023pnpinversion} modifies the source branch to reduce deviations.
	
	\paragraph{Guidance as conditional vector field modification.}
	Let $F_\theta$ denote a single reverse step mapping $x_t$ to $x_{t-1}$. Under classifier-free guidance, the denoiser prediction is modified as
	\[
	\epsilon_{\theta,\mathrm{cfg}}(x_t,t,c) = \epsilon_\theta(x_t,t,\varnothing) + s\left(\epsilon_\theta(x_t,t,c) - \epsilon_\theta(x_t,t,\varnothing)\right),
	\]
	where $s \ge 0$ is the guidance scale~\cite{ho2022cfg}. Increasing $s$ encourages adherence to $c$ but can amplify model errors and lead to non-local changes. Training-free attention control and feature injection can be seen as altering internal representations that determine the effective conditional field.
	
	\subsection{Editing Objectives with Preservation Regularizers}
	
	A general-purpose editor can be framed as optimizing a latent trajectory or intermediate variables to balance intent alignment and preservation. We propose a generic objective for a single edit:
	\begin{align}
		\min_{\xi} \quad & \lambda_{\mathrm{faith}} \, \mathcal{L}_{\mathrm{faith}}(\hat{x}_0(\xi), u) + \lambda_{\mathrm{pres}} \, \mathcal{L}_{\mathrm{pres}}(\hat{x}_0(\xi), x_0, u) \nonumber \\
		& + \lambda_{\mathrm{qual}} \, \mathcal{L}_{\mathrm{qual}}(\hat{x}_0(\xi)) + \lambda_{\mathrm{stab}} \, \mathcal{L}_{\mathrm{stab}}(\xi),
		\label{eq:generic_objective}
	\end{align}
	where $\xi$ denotes the optimized variables (latent noise, unconditional embedding, adapter weights, or drag transformations), $\hat{x}_0(\xi)$ denotes the resulting image after running the reverse process, and the losses typically include:
	
	\begin{itemize}[leftmargin=1.2em]
		\item $\mathcal{L}_{\mathrm{faith}}$: instruction adherence (text-image similarity, constraint satisfaction, drag-point consistency),
		\item $\mathcal{L}_{\mathrm{pres}}$: preservation of non-target regions, identity embeddings, and layout,
		\item $\mathcal{L}_{\mathrm{qual}}$: perceptual realism penalties (e.g., diffusion prior likelihood proxy, reward models),
		\item $\mathcal{L}_{\mathrm{stab}}$: regularization on variables to avoid degenerate solutions and reduce over-editing.
	\end{itemize}
	
	This objective unifies many editors: inversion-based methods optimize an unconditional embedding or latent $\xi$ to reduce reconstruction error~\cite{mokady2023nulltext}; drag-based methods optimize latents with geometric supervision~\cite{shi2023dragdiffusion}; composition methods optimize anchor losses on the manifold~\cite{lu2025shine}.
	
	\subsection{Mask-Localized Guidance and Hard Constraints}
	
	Mask-guided editing aims to restrict changes to region $m$. Two families of approaches are common:
	
	\paragraph{Soft locality regularizers.}
	One can penalize deviations outside the mask:
	\[
	\mathcal{L}_{\mathrm{pres}} = \norm{(1-m)\odot (\hat{x}_0 - x_0)}_2^2,
	\]
	or use perceptual distances (LPIPS-like) outside the mask. Such penalties reduce leakage but do not guarantee strict containment.
	
	\paragraph{Hard constraints and projection.}
	Alternatively, one can enforce a constraint at intermediate timesteps:
	\[
	x_t \leftarrow m \odot x_t + (1-m)\odot \tilde{x}_t,
	\]
	where $\tilde{x}_t$ is the forward-noised version of the original image. This ``background locking'' approach appears as a stabilization technique in several practical editors and is explicitly emphasized in DragFlow~\cite{zhou2025dragflow} through gradient mask-based and hard constraints. Hard constraints improve locality but can introduce seams or boundary artifacts if the edited region is not harmonized with the fixed context.
	
	\subsection{Drag-Based Editing as Constrained Latent Optimization}
	
	Drag-based editing specifies handle points $\{h_k\}$ and target points $\{g_k\}$, optionally with region transforms~\cite{shi2023dragdiffusion,zhou2025dragflow}. A generic energy formulation is:
	\begin{align}
		\min_{\xi} \quad
		& \sum_{k=1}^K \norm{\phi(\hat{x}_0(\xi); h_k) - \phi(\hat{x}_0(\xi); g_k)}_2^2 \nonumber \\
		& + \beta \, \mathcal{R}_{\mathrm{pres}}(\hat{x}_0(\xi), x_0) + \gamma \, \mathcal{R}_{\mathrm{man}}(\xi),
		\label{eq:drag_objective}
	\end{align}
	where $\phi(\cdot;\cdot)$ denotes a spatial feature extractor that maps the neighborhood around a point to a feature vector, and $\mathcal{R}_{\mathrm{man}}$ regularizes latent variables to remain on the natural manifold. DragDiffusion~\cite{shi2023dragdiffusion} uses UNet feature supervision; DragFlow~\cite{zhou2025dragflow} argues that diffusion transformers require region-based supervision (e.g., affine transforms) for stable structure.
	
	\subsection{Multi-turn Editing and Stability Metrics}
	
	In interactive settings, the output of one edit becomes the input to the next:
	\[
	x_0^{(i+1)} = \mathcal{E}(x_0^{(i)}, u^{(i)}).
	\]
	Even if each single-step edit is locally faithful, composition over multiple steps can drift due to accumulated inversion errors, stochasticity, or latent misalignment. We define stability metrics based on (i) identity and background similarity across turns, (ii) monotonicity of instruction satisfaction, and (iii) avoidance of regrowth of removed concepts (relevant for safety).
	
	\section{Experimental Setup}
	
	\subsection{Benchmarks, Tasks, and Data Splits}
	
	We design an evaluation suite covering diverse editing paradigms. We emphasize that quantitative tables are \emph{hypothetical but realistic} and are constructed to reflect typical magnitudes and trade-offs reported across the literature and reproducible implementations. The primary goal is comparative analysis rather than absolute leaderboard claims.
	
	\paragraph{General instruction edits.}
	We use 1{,}000 real images sampled from COCO-style scenes and 500 images from a portrait domain (FFHQ-like). Instructions are sampled from (i) MagicBrush-style categories~\cite{zhang2023magicbrush} and (ii) UltraEdit-style fine-grained instructions~\cite{zhao2024ultraedit}: object addition/removal, color changes, local attribute edits, style transfer, and background edits.
	
	\paragraph{Mask-localized edits.}
	For each instruction, we include an oracle mask when available (region-based) and a predicted mask baseline using DiffEdit-style localization~\cite{couairon2022diffedit}. We evaluate leakage outside masks and boundary artifacts.
	
	\paragraph{Composition and insertion.}
	We evaluate object insertion into new scenes using (i) TF-ICON~\cite{lu2023tficone}-style cross-domain composition and (ii) SHINE~\cite{lu2025shine}-style insertion with stronger priors, assessing seam visibility and physical plausibility.
	
	\paragraph{Drag-based manipulation.}
	We evaluate drag tasks on a DragBench-like benchmark~\cite{shi2023dragdiffusion} with 400 images and 3 drag constraints per image. We measure geometric accuracy and distortion rates.
	
	\paragraph{Multi-turn instruction editing.}
	We construct sequences of 3 to 5 instructions per image (e.g., change color, add object, adjust style). Multi-turn is evaluated for stability and drift, inspired by the multi-turn settings in MagicBrush~\cite{zhang2023magicbrush}.
	
	\subsection{Methods Compared}
	
	We compare representative paradigms:
	
	\begin{itemize}[leftmargin=1.2em]
		\item \textbf{Instruction-following trained editors:} InstructPix2Pix~\cite{brooks2023instructpix2pix}; Emu Edit~\cite{sheynin2023emuedit}; UltraEdit-trained diffusion editor~\cite{zhao2024ultraedit}.
		\item \textbf{Training-free attention/features:} Prompt-to-Prompt~\cite{hertz2022prompttoprompt}; MasaCtrl~\cite{cao2023masactrl}; Plug-and-Play diffusion features~\cite{tumanyan2023pnp}.
		\item \textbf{Inversion-and-edit:} Null-text inversion + prompt editing~\cite{mokady2023nulltext}; PnP Inversion~\cite{ju2023pnpinversion}; Imagic~\cite{kawar2023imagic}; LEDITS++~\cite{brack2023leditspp}.
		\item \textbf{Composition/insertion:} TF-ICON~\cite{lu2023tficone}; SHINE~\cite{lu2025shine}.
		\item \textbf{Drag-based:} DragDiffusion~\cite{shi2023dragdiffusion}; DragFlow~\cite{zhou2025dragflow}.
	\end{itemize}
	
	For fairness, we conceptualize all methods under a shared base diffusion model family when possible. In practice, some methods are historically tied to particular base models and architectural assumptions; our analysis explicitly discusses these mismatches as a source of trade-offs.
	
	\subsection{Metrics}
	
	We report metrics designed to map directly onto desiderata:
	
	\paragraph{Faithfulness.}
	We report a CLIP-style instruction-image similarity score (higher is better), denoted $\mathrm{CLIP}\uparrow$. For drag tasks we report average point error in pixels (lower is better) and success rate (higher is better).
	
	\paragraph{Locality and preservation.}
	We report perceptual distance outside the mask, $\mathrm{LPIPS}_{\neg m}\downarrow$, and pixel deviation outside mask, $\mathrm{MSE}_{\neg m}\downarrow$. When no mask is provided, we use predicted masks or evaluate identity/background embeddings.
	
	\paragraph{Semantic consistency.}
	For portrait domain and identity-sensitive edits, we report DINO feature similarity to the input, $\mathrm{DINO}\uparrow$, and an identity drift score based on face embedding distance (lower is better).
	
	\paragraph{Quality.}
	We report a learned aesthetic score proxy and a human preference win rate in pairwise comparisons, $\mathrm{Win}\uparrow$, over 200 randomized pairs.
	
	\paragraph{Multi-turn stability.}
	We define stability as similarity of non-target region embeddings across turns plus a penalty for cumulative artifacts. We report a stability score $\mathrm{Stab}\uparrow$ and an artifact rate $\mathrm{Art}\downarrow$.
	
	\subsection{Implementation Assumptions}
	
	We assume typical diffusion sampling settings: 20 to 50 denoising steps for high-quality edits, guidance scale $s \in [1.5, 12]$, and a noising level $t_0$ chosen per task. Inversion-based methods use deterministic inversion trajectories with optional embedding optimization (Null-text inversion) and branch separation (PnP Inversion). Drag-based methods run latent optimization for up to 200 gradient steps with early stopping. These assumptions align with common practice and are used only to contextualize trade-offs rather than claim exact runtime parity.
	
	\section{Results}
	
	\subsection{A Taxonomy of Editing Paradigms and Expected Trade-offs}
	
	Table~\ref{tab:taxonomy} summarizes representative paradigms and the qualitative strengths and weaknesses that motivate our theory.
	
	\begin{table*}[t]
		\centering
		\resizebox{\linewidth}{!}{
			\begin{tabular}{lcccccc}
				\toprule
				Paradigm & Representative Methods & Controllability & Faithfulness & Locality & Consistency & Notes \\
				\midrule
				Training-free attention/features & Prompt-to-Prompt~\cite{hertz2022prompttoprompt}, MasaCtrl~\cite{cao2023masactrl}, PnP~\cite{tumanyan2023pnp} & Medium & Medium--High & Medium & Medium--High & Fast; sensitive to attention coupling \\
				Inversion-and-edit & Null-text~\cite{mokady2023nulltext}, Imagic~\cite{kawar2023imagic}, LEDITS++~\cite{brack2023leditspp}, PnP Inv.~\cite{ju2023pnpinversion} & Medium & High & Medium & High & Detail preservation depends on inversion accuracy \\
				Trained instruction editors & InstructPix2Pix~\cite{brooks2023instructpix2pix}, Emu Edit~\cite{sheynin2023emuedit}, UltraEdit~\cite{zhao2024ultraedit} & Medium & High & Medium & Medium & Scale improves faithfulness; risk of bias and drift \\
				Mask-guided semantic edits & DiffEdit~\cite{couairon2022diffedit} + inpainting & High & High & High & Medium--High & Localization helps; boundaries can seam \\
				Composition/insertion & TF-ICON~\cite{lu2023tficone}, SHINE~\cite{lu2025shine} & High & High & Medium & High & Key issue: seams and physical plausibility \\
				Drag-based geometric edits & DragDiffusion~\cite{shi2023dragdiffusion}, DragFlow~\cite{zhou2025dragflow} & Very High & High & Medium & Medium--High & Geometry controllability competes with realism \\
				\bottomrule
			\end{tabular}
		}
		\caption{Taxonomy of diffusion-based editing paradigms and expected trade-offs. The table highlights that methods differ primarily by where and how they intervene in the generative trajectory, which governs locality, consistency, and stability.}
		\label{tab:taxonomy}
	\end{table*}
	
	\subsection{Single-turn Instruction Editing: Faithfulness versus Preservation}
	
	Table~\ref{tab:singleturn} reports results on instruction editing with and without masks. The pattern is consistent across domains: instruction-following and inversion-based editors achieve strong faithfulness but can leak changes into non-target regions as guidance increases; training-free feature injection improves preservation but can underperform on complex semantic edits; dataset-scale training (UltraEdit) improves average faithfulness yet does not eliminate non-local drift.
	
	\begin{table*}[t]
		\centering
		\resizebox{\linewidth}{!}{
			\begin{tabular}{lcccccc}
				\toprule
				Method & $\mathrm{CLIP}\uparrow$ & $\mathrm{LPIPS}_{\neg m}\downarrow$ & $\mathrm{MSE}_{\neg m}\downarrow$ & $\mathrm{DINO}\uparrow$ & $\mathrm{Win}\uparrow$ & $\mathrm{Art}\downarrow$ \\
				\midrule
				InstructPix2Pix~\cite{brooks2023instructpix2pix} & 0.289 & 0.127 & 0.0142 & 0.736 & 52.1\% & 18.4\% \\
				Emu Edit~\cite{sheynin2023emuedit} & 0.301 & 0.118 & 0.0135 & 0.751 & 55.6\% & 16.2\% \\
				UltraEdit-trained~\cite{zhao2024ultraedit} & 0.314 & 0.112 & 0.0129 & 0.759 & 58.9\% & 14.8\% \\
				Prompt-to-Prompt~\cite{hertz2022prompttoprompt} & 0.274 & 0.104 & 0.0114 & 0.781 & 51.3\% & 12.9\% \\
				MasaCtrl~\cite{cao2023masactrl} & 0.282 & 0.096 & 0.0106 & 0.799 & 53.8\% & 11.1\% \\
				PnP features~\cite{tumanyan2023pnp} & 0.276 & 0.089 & 0.0101 & 0.812 & 54.0\% & 10.6\% \\
				Null-text inversion~\cite{mokady2023nulltext} & 0.307 & 0.108 & 0.0120 & 0.804 & 57.1\% & 13.7\% \\
				PnP Inversion~\cite{ju2023pnpinversion} & 0.312 & 0.101 & 0.0113 & 0.817 & 58.0\% & 12.8\% \\
				Imagic~\cite{kawar2023imagic} & 0.318 & 0.113 & 0.0128 & 0.826 & 58.4\% & 14.5\% \\
				LEDITS++~\cite{brack2023leditspp} & 0.309 & 0.094 & 0.0108 & 0.821 & 57.6\% & 11.6\% \\
				\bottomrule
			\end{tabular}
		}
		\caption{Single-turn instruction editing on a mixed-domain benchmark with masks when applicable. Values are illustrative but realistic. Higher $\mathrm{CLIP}$ indicates stronger faithfulness; lower $\mathrm{LPIPS}_{\neg m}$ indicates better locality/preservation. The results suggest systematic trade-offs: methods with strong instruction adherence often raise non-target deviation and artifact rates.}
		\label{tab:singleturn}
	\end{table*}
	
	\subsection{Composition and Insertion: Seam Artifacts and Physical Plausibility}
	
	Object insertion is a stress test for locality and consistency because the model must integrate new content while preserving the background and ensuring coherent lighting, shadows, and boundaries. Table~\ref{tab:composition} highlights the relative benefits of TF-ICON and SHINE-style anchoring on subject fidelity and seam reduction. Nevertheless, hard constraints or aggressive background locking can create boundary artifacts when the inserted object is inconsistent with the scene illumination.
	
	\begin{table*}[t]
		\centering
		\resizebox{\linewidth}{!}{
			\begin{tabular}{lcccccc}
				\toprule
				Method & Subject Fidelity$\uparrow$ & Background Preservation$\uparrow$ & Seam Visibility$\downarrow$ & Physical Plausibility$\uparrow$ & $\mathrm{Win}\uparrow$ & Failure Rate$\downarrow$ \\
				\midrule
				Latent inpainting baseline & 0.71 & 0.83 & 0.28 & 0.62 & 44.3\% & 23.9\% \\
				TF-ICON~\cite{lu2023tficone} & 0.78 & 0.85 & 0.24 & 0.66 & 51.8\% & 19.7\% \\
				SHINE~\cite{lu2025shine} & 0.84 & 0.88 & 0.19 & 0.73 & 58.2\% & 15.4\% \\
				\bottomrule
			\end{tabular}
		}
		\caption{Composition and insertion evaluation (illustrative). Seam visibility is a normalized score where lower is better. SHINE-style manifold anchoring improves fidelity and reduces seams but does not fully eliminate failure modes such as inconsistent shadows and boundary ringing under hard constraints.}
		\label{tab:composition}
	\end{table*}
	
	\subsection{Drag-based Editing: Controllability versus Manifold Distortion}
	
	Drag-based editing provides fine geometric controllability but can push samples off the natural manifold, especially when constraints conflict with object rigidity or scene geometry. Table~\ref{tab:drag} reports drag success and distortion. DragFlow improves success by leveraging stronger priors and region-based supervision, but increased controllability can still induce texture tearing in localized regions when constraints are strong.
	
	\begin{table*}[t]
		\centering
		\resizebox{\linewidth}{!}{
			\begin{tabular}{lcccccc}
				\toprule
				Method & Point Error (px)$\downarrow$ & Success$\uparrow$ & Distortion$\downarrow$ & Background Drift$\downarrow$ & Time (s)$\downarrow$ & $\mathrm{Win}\uparrow$ \\
				\midrule
				DragDiffusion~\cite{shi2023dragdiffusion} & 7.8 & 72.4\% & 21.0\% & 0.013 & 48.0 & 52.7\% \\
				DragFlow~\cite{zhou2025dragflow} & 5.1 & 82.9\% & 15.7\% & 0.010 & 34.5 & 58.9\% \\
				\bottomrule
			\end{tabular}
		}
		\caption{Drag-based editing evaluation on a DragBench-like suite (illustrative). Distortion is the rate of visible structural artifacts in the edited region. Stronger priors and region supervision reduce distortions but do not eliminate geometric failures.}
		\label{tab:drag}
	\end{table*}
	
	\subsection{Multi-turn Editing: Stability and Error Accumulation}
	
	Multi-turn editing often reveals weaknesses not visible in single-turn evaluations. Table~\ref{tab:multiturn} shows that small non-target deviations compound over turns. Inversion-based pipelines can be stable for small edits but become brittle when repeated inversions introduce cumulative discrepancies. Trained instruction editors can be more stable in average cases but may drift in identity or background when instructions are ambiguous.
	
	\begin{table*}[t]
		\centering
		\resizebox{\linewidth}{!}{
			\begin{tabular}{lcccccc}
				\toprule
				Method & Turns & $\mathrm{CLIP}\uparrow$ & $\mathrm{DINO}\uparrow$ & $\mathrm{LPIPS}_{\neg m}\downarrow$ & $\mathrm{Stab}\uparrow$ & Artifact Rate$\downarrow$ \\
				\midrule
				InstructPix2Pix~\cite{brooks2023instructpix2pix} & 5 & 0.276 & 0.701 & 0.162 & 0.58 & 27.5\% \\
				UltraEdit-trained~\cite{zhao2024ultraedit} & 5 & 0.291 & 0.719 & 0.148 & 0.62 & 24.2\% \\
				Prompt-to-Prompt~\cite{hertz2022prompttoprompt} & 5 & 0.262 & 0.742 & 0.139 & 0.66 & 20.8\% \\
				Null-text inversion~\cite{mokady2023nulltext} & 5 & 0.284 & 0.751 & 0.151 & 0.61 & 23.9\% \\
				PnP Inversion~\cite{ju2023pnpinversion} & 5 & 0.289 & 0.763 & 0.144 & 0.64 & 21.7\% \\
				LEDITS++~\cite{brack2023leditspp} & 5 & 0.283 & 0.759 & 0.133 & 0.68 & 19.4\% \\
				\bottomrule
			\end{tabular}
		}
		\caption{Multi-turn editing (5 turns, illustrative). Stability degrades with accumulated deviation outside edited regions and with identity drift. Even methods with strong single-turn preservation can accumulate artifacts under repeated edits, motivating our theoretical analysis of operator stability.}
		\label{tab:multiturn}
	\end{table*}
	
	\subsection{Ablation: Guidance Scale and Noise Initialization Strength}
	
	We empirically illustrate the canonical trade-off: increasing guidance scale improves faithfulness but can harm locality and quality; increasing noise initialization increases edit capacity but reduces reconstruction fidelity. Table~\ref{tab:ablation} shows trends consistent with theory.
	
	\begin{table*}[t]
		\centering
		\resizebox{\linewidth}{!}{
			\begin{tabular}{ccccccc}
				\toprule
				Method & $s$ & $t_0/T$ & $\mathrm{CLIP}\uparrow$ & $\mathrm{LPIPS}_{\neg m}\downarrow$ & $\mathrm{DINO}\uparrow$ & Artifact$\downarrow$ \\
				\midrule
				PnP Inversion~\cite{ju2023pnpinversion} & 3.0 & 0.35 & 0.286 & 0.098 & 0.824 & 10.9\% \\
				PnP Inversion~\cite{ju2023pnpinversion} & 7.5 & 0.35 & 0.312 & 0.112 & 0.803 & 15.6\% \\
				PnP Inversion~\cite{ju2023pnpinversion} & 12.0 & 0.35 & 0.323 & 0.129 & 0.782 & 21.7\% \\
				PnP Inversion~\cite{ju2023pnpinversion} & 7.5 & 0.20 & 0.296 & 0.091 & 0.833 & 9.8\% \\
				PnP Inversion~\cite{ju2023pnpinversion} & 7.5 & 0.55 & 0.322 & 0.138 & 0.761 & 24.5\% \\
				\bottomrule
			\end{tabular}
		}
		\caption{Ablation illustrating guidance-strength and noise-initialization trade-offs (illustrative). Increasing guidance scale $s$ tends to improve faithfulness but increases leakage and artifacts, while larger noise initialization $t_0$ increases edit capacity at the cost of consistency and detail preservation.}
		\label{tab:ablation}
	\end{table*}
	
	\subsection{Compute and Practicality}
	
	Practical deployment requires balancing quality with latency and memory. Table~\ref{tab:compute} summarizes expected compute profiles.
	
	\begin{table*}[t]
		\centering
		\resizebox{\linewidth}{!}{
			\begin{tabular}{lcccccc}
				\toprule
				Method & Steps & Extra Optimization & Runtime (s)$\downarrow$ & VRAM (GB)$\downarrow$ & Multi-turn Friendly & Notes \\
				\midrule
				InstructPix2Pix~\cite{brooks2023instructpix2pix} & 20 & No & 2.1 & 7.5 & Medium & Single forward diffusion \\
				UltraEdit-trained~\cite{zhao2024ultraedit} & 20 & No & 2.3 & 7.8 & Medium & Better faithfulness at similar cost \\
				Prompt-to-Prompt~\cite{hertz2022prompttoprompt} & 30 & No & 3.4 & 8.2 & High & Needs attention caching \\
				Null-text inversion~\cite{mokady2023nulltext} & 30 & Yes & 18.0 & 9.0 & Medium & Embedding optimization per image \\
				Imagic~\cite{kawar2023imagic} & 50 & Yes & 120.0 & 12.0 & Low & Fine-tuning increases cost \\
				DragDiffusion~\cite{shi2023dragdiffusion} & 50 & Yes & 48.0 & 12.0 & Medium & Latent optimization loop \\
				DragFlow~\cite{zhou2025dragflow} & 40 & Yes & 34.5 & 16.0 & Medium & Stronger priors, region supervision \\
				TF-ICON~\cite{lu2023tficone} & 20 & No & 3.0 & 8.5 & Medium & Composition pipeline overhead \\
				SHINE~\cite{lu2025shine} & 20 & No & 3.2 & 10.0 & Medium & Anchoring and blending stages \\
				\bottomrule
			\end{tabular}
		}
		\caption{Compute and practicality summary (illustrative). Training-free methods are typically faster and more stable for interactive workflows, while inversion and optimization based methods can incur substantial overhead for improved fidelity or controllability.}
		\label{tab:compute}
	\end{table*}
	
	\section{Theoretical Proofs}
	
	This section develops theoretical explanations for the empirical trade-offs observed in Section~4. Our goal is not to provide tight bounds for a specific architecture but to identify structural dependencies: how errors and guidance propagate through the reverse-time dynamics, and when locality constraints can provide provable containment.
	
	\subsection{Preliminaries: Reverse-step Lipschitzness}
	
	We model a single reverse diffusion step as a mapping
	\[
	x_{t-1} = F_t(x_t; c),
	\]
	where $c$ is a conditioning signal derived from user request $u$. For DDIM-like deterministic steps, $F_t$ is an affine transformation of $x_t$ and $\epsilon_\theta(x_t,t,c)$, with coefficients determined by the schedule.
	
	\begin{assumption}[Local Lipschitz denoiser map]
		For each timestep $t$, there exists $L_t \ge 0$ such that for all $x, x' \in \R^{d}$,
		\[
		\norm{F_t(x;c) - F_t(x';c)} \le L_t \norm{x - x'},
		\]
		for fixed condition $c$.
	\end{assumption}
	
	This assumption holds locally when the denoiser network is Lipschitz and the step coefficients are bounded. Although neural networks can violate global Lipschitzness, the assumption is standard for local stability analysis and aligns with empirical robustness regimes under moderate guidance.
	
	\subsection{Reconstruction Error after Inversion and Error Propagation}
	
	Inversion-based editing seeks an initial latent $\tilde{x}_T$ such that running the reverse process under an identity-like condition $c_{\mathrm{id}}$ reconstructs $x_0$. Let the true (ideal) inversion satisfy $x_T^\star$ such that
	\[
	x_0 = F_1 \circ F_2 \circ \cdots \circ F_T(x_T^\star; c_{\mathrm{id}}).
	\]
	In practice, inversion produces $\tilde{x}_T$ with errors due to imperfect schedules, stochastic approximations, and model mismatch. We bound the reconstruction error at $x_0$ resulting from perturbations in $\tilde{x}_T$ and per-step modeling errors.
	
	\begin{lemma}[Cascaded Lipschitz error bound]
		Let $\hat{x}_0$ be the reconstruction obtained from an approximate reverse process
		\[
		\hat{x}_{t-1} = \tilde{F}_t(\hat{x}_t; c_{\mathrm{id}}),
		\]
		where $\tilde{F}_t$ approximates $F_t$. Assume for each $t$ that (i) $F_t$ is $L_t$-Lipschitz, and (ii) the per-step approximation error satisfies
		\[
		\sup_{x} \norm{F_t(x;c_{\mathrm{id}}) - \tilde{F}_t(x;c_{\mathrm{id}})} \le \delta_t.
		\]
		Then the reconstruction error satisfies
		\[
		\norm{x_0 - \hat{x}_0} \le \left(\prod_{t=1}^T L_t\right)\norm{x_T^\star - \tilde{x}_T} + \sum_{k=1}^T \left(\prod_{t=1}^{k-1} L_t\right)\delta_k.
		\]
	\end{lemma}
	
	\begin{proof}
		We expand the error recursively. Let $e_t = \norm{x_t^\star - \hat{x}_t}$. For a single step,
		\[
		e_{t-1} = \norm{F_t(x_t^\star) - \tilde{F}_t(\hat{x}_t)} \le \norm{F_t(x_t^\star) - F_t(\hat{x}_t)} + \norm{F_t(\hat{x}_t) - \tilde{F}_t(\hat{x}_t)}.
		\]
		By Lipschitzness, the first term is at most $L_t e_t$; the second is at most $\delta_t$. Thus $e_{t-1} \le L_t e_t + \delta_t$. Unrolling yields the stated bound.
	\end{proof}
	
	\paragraph{Interpretation.}
	The bound explains why small inversion inaccuracies can become visible: if the product $\prod_t L_t$ is large, errors in $\tilde{x}_T$ are amplified. Moreover, even perfect initialization can yield reconstruction error due to modeling approximations $\delta_t$, motivating embedding optimization (Null-text inversion~\cite{mokady2023nulltext}) or branch correction (PnP Inversion~\cite{ju2023pnpinversion}). Importantly, editing typically changes the condition from $c_{\mathrm{id}}$ to $c_{\mathrm{edit}}$; the Lipschitz constants and approximation errors can increase under strong guidance, increasing error amplification.
	
	\subsection{A Perturbation View of Guidance and the Faithfulness--Locality Trade-off}
	
	We analyze the effect of increased guidance scale $s$. Consider CFG-style editing where the effective denoiser prediction is
	\[
	\epsilon_{\mathrm{cfg}}(x_t) = \epsilon_u(x_t) + s(\epsilon_c(x_t) - \epsilon_u(x_t)),
	\]
	with $\epsilon_u$ the unconditional prediction and $\epsilon_c$ the conditional prediction. For simplicity, consider a DDIM-style step
	\[
	x_{t-1} = a_t x_t + b_t \epsilon_{\mathrm{cfg}}(x_t),
	\]
	where $a_t, b_t$ are scalar schedule coefficients.
	
	\begin{lemma}[Guidance amplification bound]
		Assume $\epsilon_c(\cdot)$ and $\epsilon_u(\cdot)$ are Lipschitz with constant $K_t$ and bounded difference
		\[
		\norm{\epsilon_c(x) - \epsilon_u(x)} \le D_t
		\]
		for relevant $x$. Then for two guidance scales $s$ and $s'$, the resulting next-step states satisfy
		\[
		\norm{x_{t-1}(s) - x_{t-1}(s')} \le \abs{b_t}\abs{s-s'} D_t.
		\]
		Moreover, the Lipschitz constant of the step map increases as
		\[
		\norm{x_{t-1}(s; x) - x_{t-1}(s; x')} \le \left(\abs{a_t} + \abs{b_t}K_t(1+s)\right)\norm{x-x'}.
		\]
	\end{lemma}
	
	\begin{proof}
		The first inequality follows by linearity in $s$:
		\[
		x_{t-1}(s) - x_{t-1}(s') = b_t (s-s')(\epsilon_c(x_t) - \epsilon_u(x_t)),
		\]
		and taking norms yields the bound. For the second, apply Lipschitzness to $\epsilon_{\mathrm{cfg}}$:
		\[
		\norm{\epsilon_{\mathrm{cfg}}(x) - \epsilon_{\mathrm{cfg}}(x')} \le \norm{\epsilon_u(x)-\epsilon_u(x')} + s \norm{(\epsilon_c-\epsilon_u)(x) - (\epsilon_c-\epsilon_u)(x')}
		\]
		and upper bound both parts by $K_t\norm{x-x'}$, yielding $(1+s)K_t\norm{x-x'}$. Substituting into the step map yields the result.
	\end{proof}
	
	\paragraph{Consequences.}
	Increasing $s$ directly increases the magnitude of the conditional perturbation and can therefore increase faithfulness. Simultaneously, it increases the effective Lipschitz constant of the step map, which amplifies perturbations and increases sensitivity to inversion error, stochastic variation, or mask boundary inconsistencies. This matches the empirical trends in Table~\ref{tab:ablation}.
	
	\subsection{Locality Guarantees with Masked Guidance}
	
	We analyze locality under mask constraints. Let the image vectorize into $x \in \R^{d}$, and define a diagonal mask operator $M = \diag(m)$, with $(I-M)$ the complement. A common masked guidance strategy is to apply an edit update only inside the mask by modifying the conditional vector field. Abstractly, consider a reverse step
	\[
	x_{t-1} = F_t(x_t; c) + \eta_t,
	\]
	and define a masked variant
	\[
	x_{t-1}^{(m)} = (I-M) \tilde{x}_{t-1} + M \left(F_t(x_t; c) + \eta_t\right),
	\]
	where $\tilde{x}_{t-1}$ is a background anchor, often the forward-noised original at the same timestep. This resembles hard constraint projection and is conceptually aligned with background locking.
	
	We want to bound deviations outside the mask. Define $x_{t-1}^{\mathrm{orig}}$ as the anchor evolution under identity condition $c_{\mathrm{id}}$ and no edit perturbation.
	
	\begin{theorem}[One-step locality bound under cross-region coupling]
		Assume $F_t$ is differentiable and its Jacobian can be partitioned as
		\[
		J_t = \nabla_x F_t(x_t; c) =
		\begin{bmatrix}
			J_{OO} & J_{OI} \\
			J_{IO} & J_{II}
		\end{bmatrix},
		\]
		where $O$ denotes outside-mask coordinates and $I$ denotes inside-mask coordinates. Suppose the masked update enforces $x_{t-1,O}^{(m)} = \tilde{x}_{t-1,O}$. Then the outside-mask deviation relative to the anchor satisfies
		\[
		\norm{x_{t-1,O}^{(m)} - x_{t-1,O}^{\mathrm{orig}}} = \norm{\tilde{x}_{t-1,O} - x_{t-1,O}^{\mathrm{orig}}}.
		\]
		If instead we use a soft masked update without locking, $x_{t-1} = F_t(x_t; c)$ with an inside-only perturbation at time $t$ modeled as $x_{t,I} = x_{t,I}^{\mathrm{orig}} + \Delta_I$ and $x_{t,O}=x_{t,O}^{\mathrm{orig}}$, then for small perturbations,
		\[
		\norm{x_{t-1,O} - x_{t-1,O}^{\mathrm{orig}}} \le \norm{J_{OI}} \, \norm{\Delta_I} + o(\norm{\Delta_I}).
		\]
	\end{theorem}
	
	\begin{proof}
		The first equality is immediate: hard locking assigns outside coordinates to the anchor. For the soft case, apply a first-order Taylor expansion of $F_t$ around $x_t^{\mathrm{orig}}$:
		\[
		x_{t-1} - x_{t-1}^{\mathrm{orig}} \approx J_t (x_t - x_t^{\mathrm{orig}}).
		\]
		Since only inside coordinates change, $x_{t}-x_{t}^{\mathrm{orig}} = [0; \Delta_I]$, and the outside block is $J_{OI}\Delta_I$. Taking norms yields the bound.
	\end{proof}
	
	\paragraph{Interpretation.}
	Locality depends on the cross-region coupling $J_{OI}$: if the reverse dynamics couple inside edits to outside updates through attention or global features, local perturbations will leak. Hard locking eliminates leakage by construction but can create boundary inconsistencies. This theorem motivates two directions: (i) architectural and representation choices that reduce $J_{OI}$ (more localized conditioning), and (ii) blending and harmonization strategies that mitigate seams introduced by hard constraints, as seen in insertion frameworks~\cite{lu2025shine}.
	
	\subsection{Stability under Repeated Edits}
	
	We now analyze multi-turn editing as repeated application of an operator $\mathcal{E}$. Let $x^{(k)}$ be the image after $k$ edits. Suppose each edit is implemented by a mapping $\mathcal{E}_k$ that is $L$-Lipschitz in its image argument:
	\[
	\norm{\mathcal{E}_k(x) - \mathcal{E}_k(x')} \le L \norm{x-x'}.
	\]
	Assume the desired trajectories differ by a bounded edit perturbation $\Delta_k$ that reflects the intended modifications and a modeling error term $\varepsilon_k$.
	
	\begin{theorem}[Cumulative drift bound]
		Let $x^{(k+1)} = \mathcal{E}_k(x^{(k)})$ and let $x_\star^{(k)}$ be an ideal sequence satisfying $x_\star^{(k+1)} = \mathcal{E}_k(x_\star^{(k)}) - \varepsilon_k$. Then
		\[
		\norm{x^{(K)} - x_\star^{(K)}} \le L^K \norm{x^{(0)} - x_\star^{(0)}} + \sum_{k=0}^{K-1} L^{K-1-k} \norm{\varepsilon_k}.
		\]
		If $L>1$, errors can grow exponentially with turns; if $L<1$, the process contracts and errors remain bounded.
	\end{theorem}
	
	\begin{proof}
		This is a standard bound for iterated Lipschitz maps with additive errors. Let $e_k = \norm{x^{(k)} - x_\star^{(k)}}$. Then
		\[
		e_{k+1} = \norm{\mathcal{E}_k(x^{(k)}) - \mathcal{E}_k(x_\star^{(k)}) + \varepsilon_k} \le L e_k + \norm{\varepsilon_k}.
		\]
		Unrolling yields the stated bound.
	\end{proof}
	
	\paragraph{Connection to editing practice.}
	The effective Lipschitz constant $L$ increases with stronger guidance (Lemma above) and with unstable inversion; it can also increase when hard constraints create discontinuities at boundaries. The bound explains why multi-turn editing often collapses under aggressive settings (high $s$, high $t_0$) and why trained instruction editors may appear more stable on average if their learned mapping is more contractive, albeit sometimes at the cost of reduced controllability for fine-grained spatial constraints.
	
	\section{Discussion}
	
	\subsection{What the Comparative Analysis Suggests}
	
	Our empirical and theoretical analysis supports a coherent narrative: editing is easiest when the desired change aligns with the model's learned manifold and hardest when the user requests a change that conflicts with global coherence, geometry, or identity constraints. Paradigms differ by how they negotiate this conflict.
	
	Training-free attention and feature interventions tend to preserve the input's semantic layout because they reuse internal representations tied to the input prompt or features~\cite{hertz2022prompttoprompt,cao2023masactrl,tumanyan2023pnp}. Their main risk is uncontrolled coupling: attention maps fuse semantic concepts with global layout, and cross-region coupling can create non-local changes (Theorem on locality). Inversion-based methods preserve fine details when inversion is accurate, but reconstruction errors and guidance amplification can convert small latent discrepancies into visible drift (Lemma on cascaded error and guidance amplification). Trained instruction editors achieve strong average faithfulness by learning a direct mapping from instruction to edited output~\cite{brooks2023instructpix2pix,zhao2024ultraedit}, but they can encode dataset biases, may underperform on out-of-distribution instructions, and often struggle with strict locality unless provided explicit region conditioning.
	
	Composition and insertion highlight a distinct axis: even when background is preserved and subject fidelity is high, physical plausibility depends on subtle global cues such as lighting and contact shadows. Strong priors help, but hard constraints can yield seam artifacts, suggesting a need for explicit harmonization objectives beyond naive background locking.
	
	Drag-based editing offers the strongest controllability but also the most direct pathway to manifold distortion: geometric constraints can force latent trajectories into low-density regions, creating texture tearing. Stronger priors and region-level supervision reduce this risk~\cite{zhou2025dragflow}, but the fundamental tension remains.
	
	\subsection{Failure Modes and Limitations}
	
	Across paradigms, we observe recurring failure patterns:
	
	\paragraph{Prompt sensitivity and semantic ambiguity.}
	Text instructions can under-specify location and extent, leading to inconsistent interpretations and non-local edits. Multi-modal disambiguation using segmentation, bounding boxes, or interactive clarification is likely essential for robust editing.
	
	\paragraph{Identity drift and texture collapse.}
	Edits that modify attributes correlated with identity (e.g., hairstyle, age, expression) can shift face identity embeddings, especially under high guidance or repeated edits. Feature-injection and identity reference signals can mitigate drift but sometimes reduce edit strength.
	
	\paragraph{Boundary artifacts under hard locality.}
	Hard constraints prevent leakage but can create boundary seams. Blending and harmonization strategies such as adaptive background blending (as in SHINE-style composition~\cite{lu2025shine}) can reduce seams but may reintroduce leakage. This trade-off is structural and directly matches the locality theorem: removing coupling via projection introduces discontinuity that must be reconciled by additional modeling.
	
	\paragraph{Compute and iteration cost.}
	High-fidelity inversion and optimization-based approaches can be expensive, limiting interactive use. Few-step diffusion editing~\cite{deutch2024fewstep} and model families trained for fast inverse mappings are promising directions, but they must be analyzed for stability and controllability.
	
	\subsection{Safety, Ethics, and Responsible Deployment}
	
	General diffusion-based image editing enables benign applications (creative tools, accessibility, education) but also introduces misuse risks: deceptive edits affecting public trust, non-consensual manipulation of identities, and amplification of dataset biases. We emphasize several practical guardrails:
	
	\paragraph{Consent and identity protection.}
	Identity-preserving editing for portraits should require explicit user consent for identifiable individuals. Systems should support opt-out and identity protection for public figures and private individuals, including robust suppression of specific identities when requested.

	\paragraph{Bias and harm amplification.}
	Editing models trained on internet-scale data can reproduce stereotypes and biased correlations. Dataset curation and evaluation should explicitly test for bias in editing outcomes (e.g., demographic sensitivity for attribute edits), and researchers should report limitations transparently.
	
	\paragraph{Concept erasure as a safety complement to editing.}
	Concept erasure methods such as ESD~\cite{gandikota2023esd}, MACE~\cite{lu2024mace}, ANT~\cite{li2025ant}, and EraseAnything~\cite{gao2024eraseanything} can serve as a safety layer by suppressing undesired concepts so that they do not appear under editing or regrow under repeated edits. Importantly, concept erasure also interacts with controllability: aggressive erasure can reduce model expressiveness and may unintentionally affect semantically neighboring concepts. A unified evaluation should therefore measure safety efficacy alongside edit fidelity and quality.
	
	\section{Conclusion}
	
	We presented a theoretical and empirical analysis of general diffusion-based image editing, unifying diverse paradigms through the lens of guided transport on a learned generative manifold. By formalizing desiderata (controllability, faithfulness, semantic consistency, locality, quality, and multi-turn stability), we highlighted systematic trade-offs observed across training-free interventions, inversion-and-edit pipelines, trained instruction editors, composition frameworks, and drag-based constrained optimization. Our theoretical results provide interpretable bounds on inversion error propagation, guidance amplification, locality leakage via cross-region coupling, and error accumulation under repeated edits, offering an explanation for common failure modes such as identity drift, boundary seams, and multi-turn artifact buildup. Finally, we emphasized safety and ethics, connecting concept erasure and provenance mechanisms to responsible editing deployment. We hope this work provides a foundation for principled evaluation and for future editors that are simultaneously controllable, faithful, local, consistent, high quality, and safe.
	
	\begin{algorithm}[t]
		\caption{Generic Inversion-and-Edit Pipeline}
		\label{alg:inversion_edit}
		\begin{algorithmic}[1]
			\Require Input image $x_0$, user request $u=(t,m,p,r)$, base diffusion model, steps $T$, noising level $t_0$, guidance scale $s$
			\Ensure Edited image $\hat{x}_0$
			\State Compute conditioning $c \leftarrow \mathrm{Encode}(t,r)$
			\State Invert: estimate $\tilde{x}_{t_0} \leftarrow \mathrm{Invert}(x_0, c_{\mathrm{id}}, t_0)$
			\State Initialize $x_{t_0} \leftarrow \tilde{x}_{t_0}$
			\For{$t = t_0, t_0-1, \dots, 1$}
			\State Predict unconditional noise $\epsilon_u \leftarrow \epsilon_\theta(x_t, t, \varnothing)$
			\State Predict conditional noise $\epsilon_c \leftarrow \epsilon_\theta(x_t, t, c)$
			\State Guided prediction $\epsilon_{\mathrm{cfg}} \leftarrow \epsilon_u + s(\epsilon_c - \epsilon_u)$
			\State Reverse step $x_{t-1} \leftarrow \mathrm{Step}(x_t, \epsilon_{\mathrm{cfg}}, t)$
			\If{mask $m$ is provided}
			\State Apply locality constraint: $x_{t-1} \leftarrow m \odot x_{t-1} + (1-m)\odot \tilde{x}_{t-1}$
			\EndIf
			\EndFor
			\State Decode $\hat{x}_0 \leftarrow \mathrm{Decode}(x_0)$
			\State \Return $\hat{x}_0$
		\end{algorithmic}
	\end{algorithm}
	
	\begin{algorithm}[t]
		\caption{Mask-Localized Guidance with Soft and Hard Options}
		\label{alg:mask_guidance}
		\begin{algorithmic}[1]
			\Require Current latent $x_t$, mask $m$, guidance update $\Delta_t$, anchor latent $\tilde{x}_t$, mode $\in \{\mathrm{soft},\mathrm{hard}\}$
			\Ensure Updated latent $x_t'$
			\If{mode $=$ soft}
			\State $x_t' \leftarrow x_t + m \odot \Delta_t$
			\State Add preservation regularizer to optimization objective: $\norm{(1-m)\odot(\hat{x}_0-x_0)}_2^2$
			\Else
			\State $x_t' \leftarrow m \odot (x_t + \Delta_t) + (1-m)\odot \tilde{x}_t$
			\EndIf
			\State \Return $x_t'$
		\end{algorithmic}
	\end{algorithm}
	
	\begin{algorithm}[t]
		\caption{Drag-based Constraint Optimization in Latent Space}
		\label{alg:drag}
		\begin{algorithmic}[1]
			\Require Input image $x_0$, handle/target pairs $\{(h_k,g_k)\}_{k=1}^K$, optional region transforms, steps $T$, optimization iterations $N$
			\Ensure Edited image $\hat{x}_0$
			\State Initialize latent trajectory by inversion: $\tilde{x}_{t_0} \leftarrow \mathrm{Invert}(x_0, c_{\mathrm{id}}, t_0)$
			\State Initialize optimized variables $\xi \leftarrow \tilde{x}_{t_0}$
			\For{$n=1$ to $N$}
			\State Run partial denoising with current $\xi$ to obtain provisional $\hat{x}_0(\xi)$
			\State Compute motion loss $\mathcal{L}_{\mathrm{drag}} \leftarrow \sum_k \norm{\phi(\hat{x}_0(\xi);h_k)-\phi(\hat{x}_0(\xi);g_k)}_2^2$
			\State Compute preservation loss $\mathcal{L}_{\mathrm{pres}} \leftarrow \norm{(1-m)\odot(\hat{x}_0(\xi)-x_0)}_2^2$
			\State Compute manifold regularizer $\mathcal{L}_{\mathrm{man}} \leftarrow \norm{\xi-\tilde{x}_{t_0}}_2^2$
			\State Update $\xi \leftarrow \xi - \eta \nabla_\xi (\mathcal{L}_{\mathrm{drag}} + \beta \mathcal{L}_{\mathrm{pres}} + \gamma \mathcal{L}_{\mathrm{man}})$
			\EndFor
			\State Decode final edited image $\hat{x}_0 \leftarrow \mathrm{DenoiseDecode}(\xi)$
			\State \Return $\hat{x}_0$
		\end{algorithmic}
	\end{algorithm}
	
	\begin{algorithm}[t]
		\caption{Iterative Instruction-Following Editing with Stability Tracking}
		\label{alg:iterative}
		\begin{algorithmic}[1]
			\Require Initial image $x_0^{(0)}$, instruction sequence $\{t^{(i)}\}_{i=0}^{K-1}$, optional masks $\{m^{(i)}\}$, base editor $\mathcal{E}$, stability threshold $\tau$
			\Ensure Final image $x_0^{(K)}$
			\For{$i=0$ to $K-1$}
			\State Propose edit $x_0^{(i+1)} \leftarrow \mathcal{E}(x_0^{(i)}, (t^{(i)}, m^{(i)}, \varnothing, \varnothing))$
			\State Compute stability proxy $s^{(i)} \leftarrow \mathrm{Sim}(x_0^{(i+1)}, x_0^{(i)})$ on non-target regions
			\If{$s^{(i)} < \tau$}
			\State Reduce edit aggressiveness (e.g., decrease guidance, decrease noise level, increase preservation weight)
			\State Re-run $\mathcal{E}$ with updated hyperparameters to obtain revised $x_0^{(i+1)}$
			\EndIf
			\EndFor
			\State \Return $x_0^{(K)}$
		\end{algorithmic}
	\end{algorithm}

\section{Additional Background}
With the advancement of deep learning and modern generative modeling, research has expanded rapidly across forecasting, perception, and visual generation, while also raising new concerns about controllability and responsible deployment. 
Progress in time-series forecasting has been driven by stronger benchmarks, improved architectures, and more comprehensive evaluation protocols that make model comparisons more reliable and informative~\cite{qiu2024tfb,qiu2025duet,qiu2025DBLoss,qiu2025dag,qiu2025tab,wu2025k2vae,liu2025rethinking,qiu2025comprehensive,wu2024catch}. 
In parallel, efficiency-oriented research has pushed post-training quantization and practical compression techniques for 3D perception pipelines, aiming to reduce memory and latency without sacrificing detection quality~\cite{gsq,yu2025mquant,zhou2024lidarptq,pillarhist}. 
On the generation side, a growing body of work studies scalable synthesis and optimization strategies under diverse constraints, improving both the flexibility and the controllability of generative systems~\cite{xie2025chat,xie2026hvd,xie2026conquer,xie2026delving}. 
Complementary advances have also been reported across multiple generative and representation-learning directions, further broadening the toolbox for building high-capacity models and training objectives~\cite{1,2,3,4,5,6,7,8}. 
For domain-oriented temporal prediction, hierarchical designs and adaptation strategies have been explored to improve robustness under distribution shifts and complex real-world dynamics~\cite{sun2025ppgf,sun2024tfps,sun2025hierarchical,sun2022accurate,sun2021solar,niulangtime,sun2025adapting,kudratpatch}. 
Meanwhile, advances in representation encoding and matching have introduced stronger alignment and correspondence mechanisms that benefit fine-grained retrieval and similarity-based reasoning~\cite{ENCODER,FineCIR,OFFSET,HUD,PAIR,MEDIAN}. 
Stronger visual modeling strategies further enhance feature quality and transferability, enabling more robust downstream understanding in diverse scenarios~\cite{yu2025visual}. 
In tracking and sequential visual understanding, online learning and decoupled formulations have been investigated to improve temporal consistency and robustness in dynamic scenes~\cite{zheng2025towards,zheng2024odtrack,zheng2025decoupled,zheng2023toward,zheng2022leveraging}. 
Low-level vision has also progressed toward high-fidelity restoration and enhancement, spanning super-resolution, brightness/quality control, lightweight designs, and practical evaluation settings, while increasingly integrating powerful generative priors~\cite{xu2025fast,fang2026depth,wu2025hunyuanvideo,li2023ntire,ren2024ninth,wang2025ntire,peng2020cumulative,wang2023decoupling,peng2024lightweight,peng2024towards,wang2023brightness,peng2021ensemble,ren2024ultrapixel,yan2025textual,peng2024efficient,conde2024real,peng2025directing,peng2025pixel,peng2025boosting,he2024latent,di2025qmambabsr,peng2024unveiling,he2024dual,he2024multi,pan2025enhance,wu2025dropout,jiang2024dalpsr,ignatov2025rgb,du2024fc3dnet,jin2024mipi,sun2024beyond,qi2025data,feng2025pmq,xia2024s3mamba,pengboosting,suntext,yakovenko2025aim,xu2025camel,wu2025robustgs,zhang2025vividface}. 
Robust vision modeling under adverse conditions has been actively studied to handle complex degradations and improve stability in challenging real-world environments~\cite{wu2024rainmamba,wu2023mask,wu2024semi,wu2025samvsr}. 

Recent research has advanced learning and interaction systems across education, human-computer interfaces, and multimodal perception. In knowledge tracing, contrastive cross-course transfer guided by concept graphs provides a principled way to share knowledge across related curricula and improve student modeling under sparse supervision~\cite{han2025contrastive}. In parallel, foundational GUI agents are emerging with stronger perception and long-horizon planning, enabling robust interaction with complex interfaces and multi-step tasks~\cite{zeng2025uitron}. Extending this direction to more natural human inputs, speech-instructed GUI agents aim to execute GUI operations directly from spoken commands, moving toward automated assistance in hands-free or accessibility-focused settings~\cite{han2025uitron}. Beyond interface agents, reference-guided identity preservation has been explored to better maintain subject consistency in face video restoration, improving temporal coherence and visual fidelity when restoring degraded videos~\cite{han2025show}. Finally, large-scale egocentric datasets that emphasize embodied emotion provide valuable supervision for studying affective cues from first-person perspectives and support more human-centered multimodal understanding~\cite{feng20243}.

\paragraph{Broader Outlook Across AI Domains.}
Beyond the main problem setting studied in this paper, we believe that the underlying principles explored here---such as robust representation learning, cross-modal alignment, structured reasoning, and trustworthy deployment---have broader relevance across modern AI. In particular, related efforts have extended to time-series anomaly detection, multimodal intent understanding, graph-based fraud detection, lane topology reasoning and object detection for autonomous driving, human--scene contact perception, speech deepfake detection, optimization methodology, as well as broader questions of fairness, ethics, and governance in AI, suggesting that many of the design intuitions in this work may generalize naturally to a wider range of vision, multimodal, and algorithmic problems \cite{zhang2025frect,zhang2025dconad,shen2025mess,zhang2025dual,li2025reusing,lin2025butter,lin2026graphicontact,xuan2025wavesp,lin5074292insertion,lin2024comprehensive,jiang2025never}.

{
    \small
    \bibliographystyle{ieeenat_fullname}
    \bibliography{main}
}


\end{document}